\documentclass[a4paper,envcountsame]{llncs}
\usepackage{comment}
\usepackage{todonotes}
\usepackage{amsmath,amssymb}
\usepackage{mathtools}
\usepackage{hyperref}
\usepackage{stmaryrd}
\usepackage{soul}
\usepackage[basic,bold]{complexity}
\usepackage[inline]{enumitem}
\newif\ifhideproofs

\ifhideproofs
\usepackage{environ}
\NewEnviron{hide}{}

\fi

\usepackage{algorithm}
\usepackage[noend]{algpseudocode}

\usepackage{listings}
\usepackage{colored-comments}

\usepackage{tikz}
\usetikzlibrary{arrows,automata}
\tikzset{
  ->,
  shorten >=1pt,shorten <=1pt,
  auto,
  node distance=1cm,
  initial text={},
  el/.style={font=\scriptsize},
  every state/.style={
    minimum size=10pt,
    inner sep=0pt,
  }
}

\makeatletter
\let\epsilon\varepsilon
\let\phi\varphi
\let\emptyset\varnothing
\makeatother

\newcommand\sem[1]{\llbracket #1 \rrbracket}
\newcommand\lcp{\text{lcp}}
\newcommand\dom{\text{dom}}

\newcommand{\calA}{\mathcal{A}}
\newcommand{\calM}{\mathcal{M}}

\renewcommand{\st}{\mid}
\newcommand{\Up}{\mathit{Up}}
\newcommand{\Aup}{\mathcal A_{\Up}}

\newcommand{\xrightarrowstar}[1]{\mathrel{
  \vphantom{\xrightarrow{#1}}
  \smash{\xrightarrow{#1}}
  \vphantom{\to}^*}
  }

\title{Active Learning of Sequential Transducers with Side Information about the Domain}
\author{
    Rapha\"el 
    Berthon\inst{1,2} \and
    Adrien 
    Boiret\inst{1} \and
    Guillermo A. P\'erez\inst{2} \and
    Jean-Fran\c{c}ois 
    Raskin\inst{1}
}
\institute{Universit\'e libre de Bruxelles, Brussels, Belgium \and
University of Antwerp -- Flanders Make, Antwerp, Belgium}

\begin{document}

\maketitle

\begin{abstract}
Active learning is a setting in which a student queries a teacher, through membership and equivalence queries, in order to learn a language.
Performance on these algorithms is often measured in the number of queries required to learn a target,
with an emphasis on costly equivalence queries.
In graybox learning, the learning process is accelerated by foreknowledge of some information on the target.
Here, we consider graybox active learning of subsequential string transducers,
where a regular overapproximation of the domain is known by the student.
We show that there exists an algorithm using string equation solvers that uses this knowledge
to learn subsequential string transducers with a better guarantee on the required number of equivalence queries
than classical active learning.
\end{abstract}

\section{Introduction}

Angluin's seminal work on (active) learning regular languages from queries and counterexamples~\cite{angluin1987query} initiated a body of works around the automated ``learning'' of 
black-box models.
Active learning 
is a way for a non-expert user to describe a formal object through behavioural examples and counterexamples,
or to obtain formal models for the behaviour of legacy or black-box systems
which can subsequently be formally verified~\cite{vaandrager17}. In this context,
additional information about black-box systems made available to the student can make learning more efficient
in practice~\cite{vaandrager-grey,benchmarks-vaandrager}.

The $L^*$ algorithm from~\cite{angluin1987query} has been extended to learn various classes of formal object,
e.g. probabilistic automata~\cite{HigueraOncina04} and, more relevant to this paper,
(subsequential deterministic) transducers on words~\cite{vilar1996query}.
%
In this work, we aim to learn transducers,
and focus on a specific class of side information:
an upper bound on the domain of the transduction.
The advantage of this \emph{graybox} model
is twofold.
First and more directly, it can be used to skip some membership queries outside the transformation's domain.
Second, by looking for transducers with the proper behaviour when limited to the upper bound,
we allow for solutions that are smaller than the canonical objects learned by $L^*$.
This, in turn, offers better guarantees than $L^*$ when we consider the number of equivalence queries required to learn a target.
This is relevant, as in cases like non-expert description or legacy-system learning,
the equivalence test is realistically unreliable, or prohibitively costly, when compared to the rest of the operations.

One motivation to focus on learning transducers, and more specifically Mealy machines, with an upper bound on the domain
comes from games.
In multi-player verification games, assumptions about other players have been proposed to facilitate strategy synthesis~\cite[for instance]{rational-synthesis,assume-admissible-synthesis}.
Such assumptions also make sense when a strategy has already been obtained (via synthesis~\cite{bcj18} or some alternative means) and one wishes to ``minimize'' it or its encoding.
A simple way to do so is to restrict the domain of the strategy to the reachable set of game configurations (under the assumptions made about the adversaries).
Finally, when the game formalism considered allows for delays or multiple choices made unilaterally by some player --- as is the case in regular infinite games~\cite{reg-inf-games} --- strategies are not implementable by Mealy machines but rather require general transducers.




\subsubsection*{Related work.}

The classical algorithm for active learning is $L^*$ 
~\cite{angluin1987query}.
It saturates a table of observations with membership queries,
then building a minimal deterministic automaton compatible with those observations to send as candidate for an equivalence query. A polynomial number of membership queries and at most $n$ equivalence queries are always sufficient to learn the automaton.

For transducers, the OSTIA algorithm~\cite{vilar1996query} generalizes $L^*$, follows a similar structure,
and offers comparable guarantees.
Like in $L^*$, the number of queries
is polynomial in the size of the minimal normal form of the target transducer.

In the case of graybox learning, the methods differ and this alters the complexity guarantees.
For instance, when learning languages from so-called ``inexperienced teachers''~\cite{leucker2012learning},
one considers a case where the teacher sometimes answers a membership query with ``I don't know''.
Under those circumstance, it is impossible to learn a unique minimal automaton.
This leads to a trade-off in complexity. On the one hand, finding the minimal automaton compatible with an incomplete
table of observations necessitates calls to \NP{} oracles (a SAT encoding is used in ~\cite{leucker2012learning}).
On the other hand, obscuring a regular language by replacing some information with ``I don't know'' will always
make the size of the minimal solution smaller or equal to the canonical minimal deterministic automaton.
When the area where the inexperienced teacher cannot answer with certainty is assumed to be regular, an $L^*$-like algorithm is provided in~\cite{leucker2012learning}.

Another work on the topic~\cite{abel2016gray} concerns Mealy machines, i.e. transducers that write one letter exactly for each letter they read.
It is shown that one can learn a composition of two Mealy machines if the first one is already known.
Just like in~\cite{leucker2012learning}, the $L^*$-type algorithm uses oracles to find minimal machines
compatible with an incomplete table of observations (as we can only know the behaviour of the second machine on the range of the first)
and offers a guarantee in the number of equivalence queries bound to the number of states
of the minimal second machine, rather than that of the composition in whole.

\subsubsection*{Contributions.}
We 
show how
to use string equations that can be encoded into SAT
to find 
a minimal transducer compatible with incomplete observations,
and to use this in an $L^*$-like algorithm. 
Our algorithm is guaranteed to issue a number of equivalence query that is bounded by
the minimal compatible transducer, rather
than the canonical one.
This difference can be a huge benefit when our upper bound is the result of known complex logical properties
or elaborate formats respected by the input, and the transformation we wish to learn is simple.

We note the differences with~\cite{leucker2012learning,abel2016gray} in objects learned, learning frameworks,
and available queries.
We focus on transducers, a class that subsumes automata and Mealy machine.
As an added benefit, transducers are as compact as automata, and as or more compact than Mealy machines they are equivalent to.
This compactness preserves or improves the equivalence queries guarantees.
In our learning framework, the upper bound is supposed to be known by the student.
This is in contrast to the inexperienced teacher case, where the scope of possible observations is unknown,
and has to be assumed regular and learned on the fly.
When it comes to available queries, \cite{leucker2012learning} assumes the student has access to containment queries
i.e. student can ask teacher if the candidates' language contains or is contained in the target, this to obtain better the guarantees.
In our model, a simple equivalence query is considered. 
Conversely, in~\cite{abel2016gray}, the only way to do a membership query is to do so on the composition of both machines.
In that regard, learning a composition is more constraining than learning with a known upper bound.
However, since finding a reverse image to an output word through a transducer is possible with good complexity,
our algorithm can be adapted to learn a composition of two transducers, in the
framework of~\cite{abel2016gray}.

\section{Preliminaries}\label{subsec:TransducersDefinition}
A \textit{(subsequential string) transducer} $\calM$ is a tuple $(\Sigma,\Gamma,Q,q_0,w_0,\delta,\delta_F)$ where $\Sigma$ is the finite input alphabet, $\Gamma$ is the finite output alphabet, $Q$ is the finite set of states, $q_0\in Q$ is the initial state, $w_0 \in \Gamma^*$ is an initial production, $\delta$ is the transition function, a partial function $Q\times\Sigma\to Q\times\Gamma^*$ and $\delta_F$ is the final function, a partial function $Q\to\Gamma^*$.
%
If $\delta(q,a)=(q',w)$ we note $q\xrightarrow{a|w}q'$. If $\delta_F(q)=w$ we say that $q$ is final, and note $q\xrightarrow{w}\top$.
We define the relation $\rightarrow^*$ by combining the input and output of several transitions: 
$\rightarrow^*$ is the smallest relation such that $q\xrightarrowstar{\varepsilon|\varepsilon}q$, and if $q\xrightarrowstar{u|w}q'$ and $q'\xrightarrow{a|w'}q''$ then $q\xrightarrowstar{ua|w\cdot w'}q''$.
We write $q_0\xrightarrowstar{u|w}q$ when $u$ reaches the state $q$ with partial output $w$. 

For every state $q \in Q$, we associate a partial function $\sem{\calM^q}(u)$ to $\calM$ from input words over $\Sigma$ to output words over $\Gamma$. Formally, \(\sem{\calM^q}(u) = w\cdot w'\) if \(q\xrightarrowstar{u|w}q_F\) and \(q_F\xrightarrow{w'}\top\) for some $q_F \in Q$ and is undefined otherwise.
Finally, we define $\sem{\calM} \coloneqq w_0\cdot\sem{\calM^{q_0}}$ and write that $\calM$ implements $\sem{\calM}$.

We write $\dom(\sem{\calM})$ to denote the domain of $\sem{\calM}$, that is the set of all $u \in \Sigma^*$ that reach a final state $q_F \in Q$. We often consider the restriction of $\sem \calM$ to a given domain $D\subseteq \Sigma^*$, and denote it $\sem{\calM}_{|D}$.




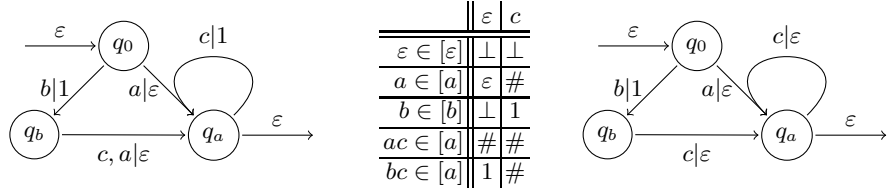
\begin{figure}
\center
\begin{minipage}{0.35\textwidth}
\center
	\begin{tikzpicture}
	\tikzstyle{state}=[draw, circle, minimum width=15pt]
		\node[state] (q0) at (0,0) {$q_0$} ;
		\node[state] (qa) [below right= of q0] {$q_a$} ;
		\node[state] (qb) [below left= of q0] {$q_b$} ;
		\node (inq0) [left= of q0] {} ;
		\node (exqa) [right= of qa] {} ;

        \draw[->] (q0) edge node[left] {$a|\varepsilon$} (qa);
        \draw[->] (q0) edge node[left] {$b|1$} (qb);
        \draw[->] (qb) edge node[below] {$c,a|\varepsilon$} (qa);
        \draw[->] (qa) edge[loop] node[above] {$c|1$} (qa);
        \draw[->] (inq0) edge node[above] {$\varepsilon$} (q0);
        \draw[->] (qa) edge node[above] {$\varepsilon$} (exqa);
 	\end{tikzpicture}
\end{minipage}
\begin{minipage}{0.25\textwidth}
\center
            \begin{tabular}{r||c|c}
             & $\varepsilon$ & $c$ \\
             \hline
             \hline
            $\varepsilon\in[\varepsilon]$ & $\bot$ & $\bot$ \\
             \hline
            $a\in[a]$ & $\varepsilon$ & $\#$ \\
             \hline
            $b\in[b]$ & $\bot$ & $1$ \\
             \hline
            $ac\in[a]$ & $\#$ & $\#$ \\
             \hline
            $bc\in[a]$ & $1$ & $\#$ \\
        \end{tabular}
\end{minipage}
\begin{minipage}{0.35\textwidth}
\center
        	\begin{tikzpicture}
	\tikzstyle{state}=[draw, circle, minimum width=15pt]
		\node[state] (q0) at (0,0) {$q_0$} ;
		\node[state] (qa) [below right= of q0] {$q_a$} ;
		\node[state] (qb) [below left= of q0] {$q_b$} ;
		\node (inq0) [left= of q0] {} ;
		\node (exqa) [right= of qa] {} ;

        \draw[->] (q0) edge node[left] {$a|\varepsilon$} (qa);
        \draw[->] (q0) edge node[left] {$b|1$} (qb);
        \draw[->] (qb) edge node[below] {$c|\varepsilon$} (qa);
        \draw[->] (qa) edge[loop] node[above] {$c|\varepsilon$} (qa);
        \draw[->] (inq0) edge node[above] {$\varepsilon$} (q0);
        \draw[->] (qa) edge node[above] {$\varepsilon$} (exqa);
 	\end{tikzpicture}
 	
\end{minipage}
    \caption{On the left, a transducer compatible with the merging map in the center, on the right the transducer resulting from this merging map. 
    }
    \label{fig:mergingnotbij}
\end{figure}

\begin{example}\label{ex:tau_abc}
Consider the function $\tau_{abc}$ with domain $\Up_{abc}=(a+bc)c^*$ and $\tau_{abc}(ac^n)=\tau_{abc}(bc^n)=1^n$.
It is implemented by the left transducer in Figure~\ref{fig:mergingnotbij}.\
\end{example}

We note that if we want to restrict a transducer's function to a regular language $L$ for which we have
a deterministic word automaton $\calA$, a classic construction is to build the product transducer $\calM\times\calA$,
where the states are the Cartesian products of both state spaces,
and the final function $\delta_F$ is only defined for pairs $(q,p)$
where $q$ is in the domain of the final function of $\calM$
and $p$ is final in $\calA$.
This transducer implements the function $\sem{\calM}_\calA$.

\paragraph{Assumptions. }
We write $|\calM|$ to denote the \emph{size} of $\calM$, i.e. its number of states.
For convenience, we only consider \emph{trim} transducers, that is to say that
every state $q$ is reachable from $q_0$ and co-reachable from a final state.

\paragraph{Active learning. }
Let $\Sigma$ and $\Gamma$ be finite input and output alphabets respectively. Further, let $\tau \colon \Sigma^* \to \Gamma^*$
be a partial function implementable by a transducer. In this work we will be interested in \emph{actively learning} a transducer implementing $\tau$ by interacting with a \emph{teacher} who knows $\tau$ and can answer questions our algorithm asks about $\tau$. Formally, the teacher is an oracle that can answer \emph{membership} and \emph{equivalence} queries.

Given $u\in\Sigma^*$, a membership query answers $\tau(u)$ if $u\in \dom(\tau)$, and $\bot$ otherwise
Given $\calM$ a transducer, an equivalence query answer \emph{true} if $\sem{\calM}=\tau$, otherwise it provides $u\in\Sigma^*$, a non-equivalence witness such that
        $u\in\dom(\sem\calM)\backslash\dom(\tau)$, or $u\in\dom(\tau)\backslash\dom(\sem\calM)$,
        or $u\in\dom(\sem\calM)\cap\dom(\tau)$ but $\sem{\calM}(u) \neq \tau(u)$.
The goal of a learning algorithm in this context is to produce a transducer $\calM$ such that $\sem{\calM} = \tau$.

\paragraph{Side information about the domain. }
We generalize the active learning problem by introducing side information available to the learning algorithm. Concretely, we assume that an \emph{upper bound} on the domain of $\tau$ is known in advance. That is, we are given a DFA $\Aup$ whose language $\Up$ is such that $\dom(\tau) \subseteq \Up$. The goal of a learning algorithm to produce a transducer $\calM$ such that $\sem{\calM}_{|\Up} = \tau$.

The domain upper bound $\Up$ may allow us to learn much simpler transducers $\calM$
than the canonical 
minimal transducer describing $\tau$ --- i.e. the class of transducers learnt by OSTIA.
For instance, if $\tau$ is the identity transformation over a complex regular language $L$,
then the canonical transducer $\calM_\tau$ will have as many states as the minimal DFA $\calA$ for $L$,
and thus will be as difficult to learn and require several equivalence tests, potentially many.
However, if $\Up=L$ is already known, then the smallest transducer $\calM$ such that $\sem{\calM}_{|\Up} = \tau$
is a simple one-state transducer, and will be identified in a single equivalence query.

\section{Learning transducers with side information}
\label{sec:TransducersUp}

Our algorithm uses an ``observation table'' $T$ based on a finite prefix-closed subset $P$ of $\Sigma^*$ and a finite suffix-closed subset $S$ of $\Sigma^*$. Formally, we define $T$ as a function $(P \cup P \cdot \Sigma) \cdot S \to \Gamma^* \cup \{\#, \bot\}$ and maintain the following invariant for all $u \in (P \cup P \cdot \Sigma)$ and all $v \in S$.
If $u \cdot v \not\in \Up$ then $T(u \cdot v) =\#$. If $u \cdot v \in \Up \setminus \dom(\tau)$ then $T(u \cdot v) =\bot$, otherwise $T(u \cdot v) =\tau(u \cdot v)$.
For technical reasons, we often consider the set $P_T$ of prefixes of the elements of $(P\cup P\Sigma)\cdot S$.


\begin{definition}[Compatible transducer]
Let 
$T$ be an observation table
and $\calM$ a transducer of input alphabet $\Sigma$ and output alphabet $\Gamma$.
We say that $\calM$ is compatible with $T$ if for all $u,v\in P\cup P\Sigma$, if $T(u\cdot v)\in\Gamma^*$ then $\sem \calM(u\cdot v)=T(u,v)$ and if $T(u\cdot v)=\bot$ then $u\cdot v\not\in\dom(\sem \calM)$.
\end{definition}

To ``fill'' the table so as to satisfy the invariant, we pose membership queries to the teacher.
Once $T$ has certain satisfactory properties (as defined in the OSTIA algorithm and elaborated upon briefly), we are able to construct a transducer $\calM$ from it.
As a table $T$ can be filled with $\#$, multiple minimal transducers may be \emph{compatible with $T$}.
To minimize the number of equivalence queries posed, we will 
send an equivalence query only if there is
a unique minimal transducer $\calM$ (up to equivalence in $\Up$) compatible with $T$. 

Instead of searching directly for transducers,
we work only with the information on how those transducers behave on $P_T$.
We represent this information using objects we call \emph{merging maps}. 
We show that we can characterize when there exist two competing minimal transducers with two different merging maps,
or two competing minimal transducers with the same merging map.
If neither is the case, then there is a unique minimal compatible transducer $\calM$, and we 
build it by guessing its merging map.
We then pose an equivalence query to the teacher in order to determine whether $\Aup\times \calM$
implements the target function $\tau$.



\paragraph{Satisfactory properties. }
The following properties are those that allow the OSTIA algorithm~\cite{vilar1996query} to work.
Under these properties, we are sure that a transducer can be derived from the table $T$.
They are defined on a specific table $T: (P\cup P\Sigma)\cdot S\rightarrow\Gamma^*\cup\lbrace\bot\rbrace$. Given $u\in P\cup P\Sigma$, we call $\lcp_T(u)$ the longest common prefix of all the $T(u\cdot v)$ in $\Gamma^*$.
For $u,u'\in P\cup P\Sigma^*$, we say that $u\equiv_{T}u'$ iff for all $v\in S$, we have both $T(u\cdot v)=\bot\iff T(u'\cdot v)=\bot$ and if $T(u\cdot v)\in\Gamma^*$ then $\lcp_T(u)^{-1}T(u \cdot v)=\lcp_T(u')^{-1}T(u'\cdot v)$.
A table $T$ is \textit{closed} if for all $ua\in P\Sigma$ there exists $u'\in P$ such that
    $ua\equiv_{T}u'$;
    \textit{$\equiv$-consistent}, if
    for all $u,u'\in P$, $a\in\Sigma$ such that $ua,u'a\in P_T$,
    then $u\equiv_{T}u'\implies ua\equiv_{T}u'a$;
    \textit{$\lcp$-consistent}, if for all $ua\in P\cup P\Sigma$, we have that $\lcp_T(u)$ is a prefix of $\lcp_T(ua)$. 

The role of these notions in Algorithm~\ref{proc:algoUpTransducer} is twofold.
First, it guarantees that the algorithm could, at worst, find the same transducer as the OSTIA algorithm~\cite{vilar1996query} as a candidate for an equivalence query
from a closed, $\equiv$-consistent, $\lcp$-consistent table.
Second, it can be seen as an efficient way to acquire information for a learner.
We can see closed, $\equiv$-consistent, $\lcp$-consistent tables as those that are saturated with membership queries,
which means that no further information can be
obtained by a learner without resorting to more costly operations,
e.g. MinGen($T$), CompetingMinGen($T$), or an equivalence query.

\paragraph{Merging maps. } 

For any given table $T$ there are infinitely many compatible transducers.
This was already the case in automata or Mealy Machines~\cite{leucker2012learning,abel2016gray}.
However, where transducers differ, is that even when limiting ourselves to transducers with a minimal number of states, this might still be the case.
Indeed, on some transitions, the output can be arbitrary (see Example~\ref{ex:MutedOpen}).
As a consequence, the method we will use to obtain a compatible transducer from a finite search space
combines the methods of~\cite{leucker2012learning} with the addition of partial output information
and an additional constraint on the output of transitions.

We want to obtain concomitantly\
an equivalence $\equiv$ on $P_T$ that describes the set of states of the transducer 
and a \emph{partial output} function $f:P_T\rightarrow\Gamma^*$
that describe which output is produced while reading an input.
In the context of transducers, side information adds another restriction:
A transducer can contain transitions that do not link together elements of $P_T$
for which we have actual output information in $T$.
This is a problem, as the output of such transitions is arbitrary and leads to an infinite number of candidates.

We will represent the behaviour of a transducer on $P_T$ but keep only the output information
that can be
corroborated in $T$.
%
We call $P_\Gamma\subseteq P_T$ the set of all $u\in P_T$
such that there exists $v\in\Sigma^*$ for which $T(u\cdot v)\in\Gamma^*$.
We call $P_\emptyset\subseteq P_T$ the set of all $u\in P_T$
such that there is no $v\in\Sigma^*$ for which $T(u\cdot v)\in\Gamma^*$.

\begin{definition}[Merging map]\label{def:MergingMap}
Let $T$ be an observation table. 
A \emph{merging map} (MM) on $T$ is a pair $(\equiv,f)$ where
$\equiv$ is an equivalence relation on $P_T$, and
$f$ is a partial function from $P_T$ to $\Gamma^*$,
such that for all $u,u'\in P_T$ and $a\in\Sigma$:

\begin{enumerate}
    \item If $f(u)$ does not exist then $u\equiv u'\ f(u')$ does not exist.
    \item If there exists $v\in \Sigma^*$ and $T(u\cdot v)\in\Gamma^*$ then $f(u)$ is a prefix of $T(u\cdot v)$.
    \item If $f(ua)$ exists, then $f(u)$ exists and is a prefix of $f(ua)$.
    \item If we have that $f(u)$ exists, $u\equiv u'$ and $ua,u'a\in P_T$ then $ua\equiv u'a$ and if $f(ua)$ exists then $f(u)^{-1}f(ua)=f(u')^{-1}f(u'a)$.
    \item If $T(u)\in\Gamma^*$, $u\equiv u'$ then $T(u')\neq\bot$ and if $T(u')\in\Gamma^*$ then $f(u)^{-1}T(u)=f(u')^{-1}T(u')$.
    \item If $f(ua)$ exists, there is no $v\in P_T$ such that $v\equiv u$, and $va\in P_\Gamma$,
    then $f(ua)=f(u)$.
\end{enumerate}
\end{definition}

The idea behind this definition is that a MM $(\equiv,f)$ contains the information necessary
to model the behaviour on $P_T$ of a transducer compatible with $T$,
while rule $6$ ensures we only keep output information from $P_\Gamma$.
If such a couple $(u,a)$ exists, we say that it is \emph{muted}.
Notably, every transducer $\calM$ compatible with $T$ has an underlying MM $(\equiv,f)$,
and conversely, every MM $(\equiv,f)$ can be used to build a transducer $\calM$ compatible with $T$.
The size of a MM is the number of equivalence classes of $\equiv$ in $\dom(f)$.

\begin{definition}[Resulting Transducer]
\label{def:restr}
Let $T$ be an observation table 
and $(\equiv,f)$ a MM on $T$.
In the transducer $\calM$ \emph{resulting from} $(\equiv,f)$ the set of states is the set of equivalence classes of $\equiv$ in $\dom(f)$ (For convenience, for $u\in P_T$ we call $q_u$ the state associated to it), the initial state is $q_\varepsilon$, the initial production is $f(\varepsilon)$, the transitions are $q_u\xrightarrow{a|f(u)^{-1}f(ua)}q_{ua}$ for $u,ua\in\dom(f)$, and for each $u$ such that $T(u)\in\Gamma^*$, we have $\delta_F(q_u)=f(u)^{-1}T(u)$

\end{definition}

\begin{definition}[Induced MM]
\label{def:InducedMergingMap}
Let $T$ be an observation table 
and $\calM$ a transducer compatible with $T$.
The MM $(\equiv,f)$ \emph{induced by the transducer} $\calM$ is such that we have
\begin{enumerate*}[label=(\Alph*)]
    \item $u\equiv v$ iff $u$ and $v$ reach the same state of $\calM$;
    \item for all $u\in P_T$, $a\in\Sigma$ such that $ua\in P_T$ reaches a state $q$ of $\calM$:
    \begin{enumerate*}[label=(B.\Roman*)]
        \item if there exists $v\in P_T$ such that $v\equiv u$, and $va\in P_\Gamma$, then
        $f(ua)=f(u)\cdot\delta(q,a)$
        \item and if $(u,a)$ is muted, then $f(ua)=f(u)$.
    \end{enumerate*}
\end{enumerate*}

\end{definition}

As an important note, these transformations are not one-to-one and do not establish a bijection
between transducers and MMs.
Notably, some transducers compatible with $T$ cannot be obtained with this method.
For instance, let us consider a table full of $\#$.
Since no $T(u)$ is ever in $\Gamma^*$, there is no final state in any transducer created with this method.

Rather than a problem, this is the goal of projecting the transducers' behaviour on $P_T$,
as we naturally eliminate candidates that are needlessly elaborate outside of the observations of $T$.
The MM induced by $\calM$ only contains information on its behaviour on $P_T$,
and the transducer resulting from a MM is the transducer with the smallest amount of states and transitions
whose behaviour on $P_T$ matches what is described in the MM.

\paragraph{Learning algorithm. } Our learning algorithm works as follows:
\begin{enumerate*}[label=(\arabic*)]
    \item We build up $T$ until it is closed and $\equiv$ and $\lcp$-consistent.
    \item If two minimal compatible transducers exist, we
    find them and a word $u$ to tell them apart.
    We use a membership query on $u$ and start again.
    \item If only one minimal compatible transducer $\calM$ remains, we
    find it.
    We use an equivalence query on $\Aup\times\calM$.
\end{enumerate*}
\begin{algorithm}
\caption{}
\begin{algorithmic}[1]
    \State Let $P=S=\{\varepsilon\}$
    \While{True}
    \If{$(u,a,v,v')$ is a witness of non-$\lcp$-consistency}
        add $av,av'$ to $S$
    \ElsIf{$(u,u',a,v)$ is a witness of non-$\equiv$-consistency}
        add $av$ to $S$
    \ElsIf{$ua$ is a witness of non-closure}
        add $ua$ to $P$
    \ElsIf{$u$ := CompetingMinGen($T$) $\neq \emptyset$}
        add $u$ and its suffixes to $S$
    \Else{ }
        $\calM$ := MinGen($T$)
        \If{$u$ is a non-equiv. witness for $\Aup\times\calM$}
            add all its suffixes to $S$
        \Else{ }
            return $\calM$
        \EndIf
    \EndIf
    \EndWhile
\end{algorithmic}
\caption{MinTransducerUp($L$,$\Aup$)}
\label{proc:algoUpTransducer}
\end{algorithm}
Such an algorithm allows using the knowledge of $\Up$ to propose
more compact candidates, as the minimal transducer compatible with a table $T$ is always smaller
than the canonical transducer that can be derived from $T$ if we substitute $\bot$ for the $\#$.
This smaller model size leads to a better guarantee when it comes to the number of required equivalence queries. The full algorithm is in Algorithm~\ref{proc:algoUpTransducer}. It uses the subprocedures CompetingMingGen and MinGen which we elaborate upon later.

\begin{theorem}\label{thm:guar-tr}
    Algorithm~\ref{proc:algoUpTransducer} terminates and makes a number of equivalence queries bounded by the number of states of a minimal $\calM$ such that $\sem{\calM}_{|\Up} = \tau$.
\end{theorem}

\begin{proof}[Sketch]
We first assume termination and focus on the bound on equivalence queries. Note that, by construction of the tables, any minimal $\calM$ such that $\sem{\calM}_{|\Up} = \tau$ is compatible with all of them. Thus, it suffices to argue that every equivalence query our algorithm poses increases the size of a minimal transducer compatible with it.
For termination, it remains to bound the number of membership queries and calls to the subprocedures. Unlike for automata, it is impossible to enumerate all transducers with $n$ states compatible with an observation table. Termination will follow from the fact that we do not consider the all such transducers. Rather, we enumerate a finite subset of them (see Prop.~\ref{prop:AlgoSizenFinite}).\qed
\end{proof}


%

%

\section{Merging maps to guess a minimal transducer}\label{sec:TransducersUpOracles}
Algorithm~\ref{proc:algoUpTransducer} relies on the two subprocedures CompetingMinGen(T) and MinGen(T)
to find one or several competing transducers compatible with an observation table.
This type of procedures is absent from blackbox learning algorithms,
but central to graybox learning algorithm~\cite{leucker2012learning}.
In the automata case, an oracle that guesses a minimal compatible automaton needs only to guess an equivalence relation on $P_T$
For transducers, we guess a function $f$ that associates to each element of $P_T$ an output in $\Gamma^*$.
Since this is not a finite search space, we aim to restrict ourselves to a finite subspace
that still allows us to find one unique or two non-equivalent minimal candidates,
when they exist.
We will limit the scope of this search with Definition~\ref{def:MutedTransition} and~\ref{def:OpenTransition}
of \emph{muted} and \emph{open} transitions, to fix arbitrary outputs at $\varepsilon$.

To combine the two subprocedures CompetingMinGen($T$) and MinGen($T$),
we characterize a necessary and sufficient condition under which there exist two possible minimal candidates,
tested by CompetingMinGen($T$).
When we ensure that the minimal candidate is unique up to equivalence on $\Up$,
we use MinGen($T$) to generate it, then send an equivalence query.
%

\paragraph{MinGen(T) using MMs.}
Recall that there are transducers compatible with a table $T$ that do not result from a MM on $T$.
We will show that to implement MinGen($T$) and CompetingMinGen($T$),
it is enough to focus on minimal MMs
and to take the resulting transducers as candidates.
To justify that this method provides the right result, we prove that it provides valid candidates.
%
%
%
%
%
%
%
%
\begin{lemma}\label{lem:resultingtransduceralignment}
Let $(f,\equiv)$ be a minimal MM on a table $T$ and $\calM$ its resulting transducer. Then, $\calM$ is compatible with $T$.
\end{lemma}
%
Among the minimal transducers compatible with $T$, there is one resulting from a MM.
Indeed, from a transducer $\calM$ compatible with $T$ one can create a smaller one using the MM induced by $\calM$ and Definition~\ref{def:restr}.
%
\begin{proposition}
\label{prop:MinimalCorrespondingTable}
Let $T$ be a table, $\calM$ a transducer compatible with $T$. There is a transducer $\calM'$, with as many states, compatible with $T$ resulting from a MM.
\end{proposition}

\paragraph{CompetingMinGen(T) using MMs.} 
While guessing a MM is enough to guess a minimal transducer, it does not provide a reliable
way to decide whether two non-equivalent minimal compatible transducers exist.
For the subroutine CompetingMinGen($T$), we must find a way to detect whether this is the case.
A natural first step is to say that if we can find minimal MMs
whose resulting transducers are non-equivalent on $\Up$, then we have found a solution to CompetingMinGen($T$). Unfortunately, this condition is not necessary.
%
Indeed, 
there are minimal MM
induced by several non-equivalent transducers.
This 
arises when a transition going out of the state associated to some $u\in P_T$
can have an arbitrarily defined output, because $ua\in P_\emptyset$, or $ua\not\in P_T$.

\begin{example}\label{ex:MutedOpen}
In Figure~\ref{fig:mergingnotbij}, we note the special case of two transitions in the left transducer: the transition $q_a\xrightarrow{c|1}q_a$ linking $a\in P_\Gamma$ to $ac\in P_\varepsilon$, and the transition $q_b\xrightarrow{a|\varepsilon}q_a$ linking $b\in P_\Gamma$ to $ba\notin P_T$.
In both cases, the transition is never used by any $u\in P_T$ such that $T(u)\in\Gamma^*$.
In other words, those transitions could be deleted, or their output arbitrarily changed for any $w\in\Gamma^*$,
without breaking compatibility.
The right transducer is also compatible with $T$, but the output of
$q_a\xrightarrow{c|1}q_a$ 
has been changed to $\varepsilon$
and the transition $q_b\xrightarrow{a|\varepsilon}q_a$ has been deleted.
\end{example}

The first case, $ua\in P_\emptyset$, is the one we aimed to eliminate by erasing the output in muted pairs $(u,a)$.
We call muted transitions those whose output has to be $\varepsilon$
in a transducer induced from a MM.
\begin{definition}\label{def:MutedTransition}
Let $T$ be a table, $(\equiv,f)$ a MM, and $\calM$ its resulting transducer.
For all $u\in P_T$, $a\in\Sigma$, $(u,a)$ is a \emph{muted pair} of $(\equiv,f)$,
and $q_u\xrightarrow{a|\varepsilon}q_{ua}$ is a \emph{muted transition} of $\calM$,
if $u,ua\in\dom(f)$ but
there is no $v\in P_T$ such that $u\equiv v$ and $va\in P_\Gamma$.
\end{definition}

The second case, $ua\not\in P_T$, is new. We formalize this notion
as open ends on a MM.
An open end is a place where a transition could be added
without influencing the behaviour of the resulting transducer on $P_T$.
We decide to fix the output of such transitions to $\varepsilon$. 
%
%
\begin{definition}\label{def:OpenTransition}
Let $T$ be a table and $(\equiv,f)$ a MM.
For all $u\in P_T$, $a\in\Sigma$, $(u,a)$ is an \emph{open end} of the map if
there exists no $v\in P_T$ such that $v\equiv u$ and $va\in P_T$.

Let $\calM$ be the resulting transducer of $(\equiv,f)$.
We say that $\calM'$ is an \emph{open completion} of $(\equiv,f)$ (or of $\calM$)
if it is the transducer $\calM$ with at most one additional transition
$u\xrightarrow{a|\varepsilon}u'$ per open end $(u,a)$.
We call such transitions \emph{open transitions}. 
\end{definition}

Muted and open transitions allow arbitrary output,
meaning there can exist several non-equivalent compatible transducers:
if there exists a word $u\in\Up$ that goes through a muted transition,
that is sufficient to build several compatible transducers that give different outputs on $u$.
This condition together with the existence of competing minimal MMs
give a necessary, sufficient and effective, condition for CompetingMinGen($T$).

%

\begin{lemma}\label{lem:ArbitraryTransducer}
Let $T$ be an observation table, $(\equiv,f)$ a MM on $T$ and $\calM$ its resulting transducer.
If there exists an open completion $\calM'$ and an element $u\in\Up$ such that
$u\in\dom(\sem{\calM'})$ and $u$ uses a muted or open transition in its run in $\calM'$,
then there exist competing minimal transducers compatible with $T$.
\end{lemma}

%
\paragraph{Implementation:} We prove that the following is a possible implementation of CompetingMinGen($T$):
\begin{enumerate*}[label=(\arabic*)]
    \item search for two minimal MMs with non-equivalent corresponding transducers,
    \item if these do not exist, search for a minimal MM and an open completion as in
    Lemma~\ref{lem:ArbitraryTransducer};
    \item otherwise, we have a unique minimal transducer up to equivalence on $\Up$.
\end{enumerate*}

\begin{proposition}\label{prop:CompetingNSC}
Let $T$ be a table. 
If there exist two minimal transducers $\calM_1$ and $\calM_2$ compatible with $T$ but not equivalent on $\Up$,
one of the following exists:
\begin{enumerate*}[label=(\roman*)]
    \item two minimal MMs with non-equivalent resulting transducers $\calM'_1,\calM'_2$, or
    \item an open completion $\calM'$ of a minimal MM compatible with $T$
    and a word $u\in\dom(\sem{\calM'})\cap\Up$ using at least one open or muted transition of $\calM'$.
\end{enumerate*}
\end{proposition}

\section{Encoding into string equations}

To align our result with other graybox algorithms~\cite{leucker2012learning,abel2016gray},
we wish to encode the minimal generation subroutines into an \NP{} problem like SAT.
While a direct encoding is possible, it is easier to go through a first encoding into string equations.
We only use operations that are easily encoded into SAT: 
word equality and inequality, concatenation, first order operators,
and a restricted use of quantifiers. Universal quantifiers are limited to 
sets of polynomial size,
and as such are a shortcut for a conjunction,
and existential quantifiers are limited to 
words of polynomial size, that can be encoded using a polynomial number of variables.

This setting has the advantage of being more directly relevant to the notions we consider,
while keeping the \NP{} theoretical bound.
Furthermore, SMT solvers have specialized tools~\cite{zheng2013z3,liang2016efficient} to solve such equations,
that may yield better practical results than direct SAT encoding.

We encode an observation table $T$, merging maps $(\equiv,f)$,
and runs of $u\in\Up$ with output $w\in\Gamma^*$ in the resulting transducer of $T$.
We use word variables $T_u$ for $T(u)$,
booleans $E_{u,v}$ for $u\equiv v$,
word variables $f_u$ for $f(u)$,
word variable $u$ and letter variables $a_i\in[1,k]$ with $u=a_1\cdot \dots \cdot a_k$ for an input word of length $k$,
and word variables $w=w_0\cdot w_1\cdot\dots\cdot w_k\cdot w_{k+1}$ for their output step by step.
The bounds on the size of $u$ is given by small model theorems in automata and transducers.

We use string equation formulae to encode the properties we combine in the minimal generation subroutines.
We classically build $\phi_{eq}$ that ensures $E_{u,v}$ denotes an equivalence.
Then, each point of Definition~\ref{def:MergingMap} can be seen as a simple combination of string equations
on $T_u$ and $f_u$ using the binary variables $E_{u,v}$ for $u,v\in P_T$.
We can thus build $\phi_{mm}$ that ensures $E_{u,v}$ and $f_u$ denote a MM.

For the transducer resulting from $(\equiv,f)$, and its open completions, we add booleans $m_{u,a}$,  $o_{u,a}$
that indicate leaving $q_u$ with $a$ is a muted or open transition.

To model runs, we use $\phi_{run}(u,w)$ ensuring $u$ has a run with output $w$
in the transducer resulting from $E_{u,v}$ and $f_u$.
We build it by making sure the run starts in the initial state with production $w_0$,
ends in a final state with production $w_{k+1}$, and at the $i^\text{th}$ step,
the input letter is $a_i$ and the output word is $w_i$.

To encode MinGen($T$) we only need to find $E_{u,v},f_u$ that respect $\phi_{mm}$ with $n$ states,
where $n$ starts at $1$ and increases until a solution is found.

For CompetingMinGen($T$) we use Proposition~\ref{prop:CompetingNSC} to split the encoding in two.
To encode the case where there exist two non-equivalent MMs,
we use variables $E_{u,v}$ and $f_u$ respecting $\phi_{mm}$ for a first MM, 
copies $E'_{u,v}$ and $f'_u$ respecting $\phi_{mm}$ for a second MM, 
and $\phi_{\Up}$ and $\phi_{run}$ to encode the existence of $u\in\Up$
whose run differs in the transducers resulting from both MMs.

It is easy to encode the case where there exist an open completion and a word $u\in\Up$ that uses an open or muted transition,
by using $m_{u,a}$ and $o_{u,a}$ on the run of $u$ in the transducer resulting from the MM of $E_{u,v}$ and $f_u$.

Combined together, they encode the minimal generating subroutines in string equations,
that could then be encoded in SAT, leading to our result:
\begin{proposition}\label{prop:EncodeMinimalGenerators}
    Let $T$ be an observation table.
    The subroutines MinGen($T$) and CompetingMinGen($T$) can be effectively implemented.
\end{proposition}

\paragraph{Note on complexity:} As this string-constraints encoding is a polynomial shortcut for a SAT encoding, each oracle call 
solves an \NP{} problem.
Coarsely, MinGen($T$) and CompetingMinGen($T$) are of complexity $ \P^{\NP}$.
To find a minimal MM of size $n$, we need $n-1$ of those oracles to fail on sizes $1\leq i < n$.

If we take Algorithm~\ref{proc:algoUpTransducer} in its entirety, then it is not each call
to MinGen($T$) and CompetingMinGen($T$) that needs $n$ of those oracles to fail to find a minimum of size $n$,
as we can keep the memory of the current minimal size for future calls.
\section{Conclusion}

Adapting graybox learning to transducers revealed more complex than expected. Our solution relies on merging maps, muted and open transitions while offering
better bounds on equivalence queries than OSTIA.
%
Two main questions remain open:
\begin{enumerate*}[label=(\arabic*)]
\item The bound on the number of equivalence queries was the aim of this paper, but the number of membership queries
or call to string equations solvers are not considered.
Providing tight bounds or proposing a potential tradeoff, like the one described in~\cite{abel2016gray},
would increase the viability of the implementation of such an algorithm.
\item We could consider other classes of  side information like
general upper bound that cut sections of $\Sigma^*\times\Gamma^*$.
\end{enumerate*}

\bibliographystyle{abbrv}
\bibliography{papers}
\clearpage
\newpage

\appendix

\begin{center}
{\Large Appendix}
\end{center}

\section{Proof of Theorem~\ref{thm:guar-tr}}
\begin{proof}
We first assume termination and focus on the bound on equivalence queries. Note that, by construction of the tables, any minimal $\calM$ such that $\sem{\calM}_{|\Up} = \tau$ is compatible with all of them. Thus, it suffices to argue that every equivalence query our algorithm poses increases the size of a minimal transducer compatible with it.
Let $T$ be the table and $\calM$ be the input transducer in an equivalence query, and $T'$ the table and $\calM'$ the transducer in the next one. We prove that $\calM'$ has at least one more state than $\calM$.

We know two things about $\calM'$.
First, it is compatible with $T$
as every observation of $T$ is still present in $T'$ and $\calM'$ is compatible with $T'$.
Second, it is not equivalent to $\calM$ on $\Up$ since $\calM$ was sent as candidate for an equivalence query that failed. This means the teacher provided a counterexample $u\in\Up$ that $\calM$ disagreed with.
This new $u$ is added to the table to obtain $T'$. Hence, $\calM$ is not compatible with $T'$, while $\calM'$ is and they differ at least on $u\in\Up$.

We now prove that $\calM'$ has strictly more states than $\calM$.
Since $\calM$ was sent as candidate for an equivalence query, that means
CompetingMinGen($T$) found that all minimal transducers compatible with $T$ were equivalent on $\Up$,
and MinGen($T$) found that $\calM$ was a minimal transducer compatible with $T$.
We know that $\calM'$ is compatible with $T$ but not equivalent with $\calM$.
Hence, it is not a minimal transducer compatible with $T$, which means it has strictly more states.

For termination, it remains to bound the number of membership queries and calls to the subprocedures. Note that, unlike in automata, it is impossible to enumerate all transducers with $n$ states or less compatible with an observation table. Termination will actually follow from the fact that we do note consider the set of all such transducers. Rather, we enumerate a useful finite subset of them,
as we end up showing in Proposition~\ref{prop:AlgoSizenFinite}.\qed
\end{proof}

\section{Proof of Lemma~\ref{lem:resultingtransduceralignment}}

\begin{proof}
First, it can be proven by simple induction that
for $u\notin\dom(f)$, $u$ has no run in $\calM$,
and for $u\in\dom(f)$, $q_\varepsilon\xrightarrow{u|f(u)} q_u$ in $\calM$.
Then, we finish with the final output in $\calM$.
If $T(u)\in\Gamma^*$ then $q_u$ is final in $\calM$
with a final output of $f(u)^{-1}T(u)$.
This means that $\sem \calM(u)=f(u)\cdot f(u)^{-1}T(u)=T(u)$.
Conversely, if $T(u)=\bot$, then either $u\notin\dom(f)$ and $u$ has no run in $\calM$,
or $u\in\dom(f)$, and by point $5$ of Definition~\ref{def:MergingMap}
\rbchanged{we obtain that for all $u'\in P_T$ such that $u'\equiv u$ we have $T(u')=\bot$.}
This means that $q_u$ is not final and thus $\sem\calM(u)$ is undefined.\qed
\end{proof}

\section{Proof of Proposition~\ref{prop:MinimalCorrespondingTable}}

\begin{proof}
We consider the equivalence $\equiv_{\calM}$ and function $f_{\calM}$ induced by Definition~\ref{def:InducedMergingMap} on $\calM$.

We show that ($\equiv_{\calM}$,$f_{\calM}$) is a merging map:

\begin{enumerate}
    \item for $u,u'\in P_T$, if $f_{\calM}(u)$ does not exist then $u$ is not in $\dom(\sem \calM)$.
    Thus, $u\equiv_{\calM} u'\iff u'\not\in\dom(\sem \calM)\iff f(u')$ does not exist
    \item if there exists $v$ such that $T(u\cdot v)\in\Gamma^*$,
    then as $\calM$ is compatible with $T$ we have that $\sem \calM(u\cdot v)=T(u\cdot v)$.
    Thus $u$ has a partial output in $\calM$, this is $f_{\calM}(u)$, and it is a prefix of $T(u\cdot v)$.
    \item for all $u\in P_T$ and $a\in\Sigma$, if $f_{\calM}(ua)$ exists, then it is either $f_{\calM}(u)$ or $f_{\calM}(u)\cdot\delta(q,a)$.
    In both cases $f_{\calM}(u)$ is a prefix of $f_{\calM}(ua)$.
    \item for all $u\in P_T$ if $f_{\calM}(u)$ exists then for all $u'\in P_T$ and $a\in\Sigma$
    such that $u\equiv_{\calM} u'$ and $ua,u'a\in P_T$,
    in $\calM$ the transition linking the state of $u$ to the state of $ua$ is the same as
    the one linking the state of $u'$ to the state of $u'a$.
    The output on that transition is the difference between the partial output of $u$ and $ua$,
    and the same difference between the partial output of $u$ and $ua$.
    In other word, it is $f_{\calM}(u)^{-1}f(ua)$ and the identical $f_{\calM}(u')^{-1}f_{\calM}(u'a)$.
    \item for all $u\in P_T$, if $T(u)\in\Gamma^*$, then $u\in\dom(\sem \calM)$.
    Hence if $u\equiv_{\calM} u'$, then $u'\in\dom(\sem \calM)$.
    Since $\calM$ is compatible with $T$, that means $T(u')\neq\bot$.
    Furthermore, if $T(u')\in\Gamma^*$, then $f_{\calM}(u')^{-1}T(u')$ is whatever output is left to produce $\sem \calM(u')$
    from the partial output of $u'$. This is the image of the state of $u'$ by $\delta_F$.
    Since this is also the case for $u$ we get that $f_{\calM}(u)^{-1}T(u)=f_{\calM}(u')^{-1}T(u')$
    \item for all $u\in P_T$ and $a\in \Sigma$, if $(u,a)$ is muted, then $f_{\calM}(ua)=f(u)$.
\end{enumerate}
By Lemma~\ref{lem:resultingtransduceralignment}, the transducer $\calM'$ obtained by using the merging map ($\equiv_{\calM}$,$f_{\calM}$) in Definition~\ref{def:restr} is compatible with $T$.
The number of equivalence classes in $\equiv_{\calM}$ (and thus the number of states of the resulting transducer)
is equal to or smaller than the number of states of $\calM$.\qed 
\end{proof}

\section{Muted Transitions and Open Ends}
We will presently prove Lemma~\ref{lem:ArbitraryTransducer}. Before that, though, we first establish some intermediate results.

\begin{lemma}\label{lem:MutedTransition}
Let $T$ be a table, $(\equiv,f)$ a merging map,
$\calM$ its resulting transducer, and $(u,a)$ a muted transition.
\rbchanged{Deleting $u\xrightarrow{a|\varepsilon}u$
or for any $u'\in P_T, w\in\Gamma^*$ replacing it by $u\xrightarrow{a|w}u'$ in $\calM$}
creates a transducer $\calM'$ compatible with $T$.
\end{lemma}

\begin{proof}[of Lemma~\ref{lem:MutedTransition}]
\rbchanged{If a muted transition $(u,a)$ exists, it is never used in the run of any element $v\in P_T$ such that $T(v)\in\Gamma^*$.
If it was, that would mean that for some $v_{pref}\in P_T,v_{suf}\in\Sigma^*$ we have $v=v_{pref}\cdot a\cdot v_{suf}$ where $a$ is the letter that uses the muted transition.
We would then have $v_{pref}\equiv u$,
but also that $v_{pref}\in P_T\backslash P_\emptyset$ as $T(v)\in\Gamma^*$,}
which contradicts the fact that $(u,a)$ is a muted transition.

Since the muted transition is not used by any element $v\in P_T$ such that $T(v)\in\Gamma^*$,
it can be deleted or its output arbitrarily modified without impacting the compatibility of the transducer.\qed
\end{proof}

\begin{lemma}\label{lem:OpenTransition}
Let $T$ be a table, $(\equiv,f)$ a merging map,
$\calM$ its resulting transducer, and $(u,a)$ an open end for $T$.
For all $u'\in\dom(f),w\in\Gamma^*$, adding $u\xrightarrow{a|w}u'$ to $\calM$
creates a transducer $\calM'$ compatible with $T$.
\end{lemma}
\begin{proof}[of Lemma~\ref{lem:OpenTransition}]
It is possible to add this transition as no transition outgoing $u$ reading $a$ exists in $\calM$.
If it did, it would have been created by a couple $v,va\in P_T$ where $v\equiv u$.
Since $(u,a)$ is an open end, no such $v,va$ exist.

This transition does not break compatibility with $T$, as no element of $P_T$ will go through
this transition during a run.
If such an element $va\in P_T$ used this transition as its last step, this would provide
$v\in P_T$ such that $v$ reaches the state of $u$, which means $v\equiv u$.
Since $(u,a)$ is an open end, no such $v,va$ exist.\qed
\end{proof}

We are now ready to give the announced proof.
\begin{proof}[of Lemma~\ref{lem:ArbitraryTransducer}]
$\calM'$ has the same number of states as $\calM$, as we just add transitions.
Lemma~\ref{lem:OpenTransition} guarantees that $\calM'$ is compatible with $T$.
Conversely, Lemma~\ref{lem:MutedTransition} guarantees that
removing the muted transitions of $\calM$ would give another transducer $\calM''$ compatible with $T$.

Since $u$ uses an open transition in its run in $\calM'$, it would fail a run in $\calM''$,
where this transition would not exist. From this we deduce $u\not\in\dom(\sem {\calM''})$.
Since $u\in\dom(\sem{\calM'})$, and $u\in\Up$ we have $\sem{\calM'}_{|\Up}\neq\sem {\calM''}_{|\Up}$.
Hence, $\calM'$ and $\calM''$ are competing minimal transducers compatible with $T$.\qed
\end{proof}

\section{Proof of Proposition~\ref{prop:CompetingNSC}}

\begin{proof}
We take $\calM_1,\calM_2$ and consider their induced merging map.
This merging map cannot have less states as $\calM_1,\calM_2$ are minimal.
We call $\calM'_1,\calM'_2$ its resulting transducers. We note that they have the same states than $\calM_1,\calM_2$
but some of the transitions' outputs might have been deleted.

If $\calM'_1$ and $\calM'_2$ are not equivalent on $\Up$, then the first statement is true,
as both $\calM'_1,\calM'_2$ are minimal and come from merging maps.

If $\calM'_1$ and $\calM'_2$ are equivalent on $\Up$, then either $\calM_1$ and $\calM'_1$ are not equivalent on $\Up$,
or $\calM_2$ and $\calM'_2$ are not equivalent on $\Up$.
We call $\calM,\calM'$ this non-equivalent pair.
Since $\calM'$ was built as the resulting transducer of the induced map of $\calM$,
we have that $\calM'$ is identical to $\calM$ for all transitions except the transitions coming from open ends
(that are not built by Definition~\ref{def:restr}),
and the outputs of muted transitions. 
Lemmas~\ref{lem:OpenTransition} and~\ref{lem:MutedTransition} ensure that we can thus put the output
of all those transitions at $\varepsilon$ to change $\calM$ into $\calM''$, a minimal transducer compatible with $T$
and an open completion of the transducer $\calM'$.

Since $\calM$ and $\calM'$ are not equivalent on $\Up$,
we consider $u\in\Up$ a counterexample.
$\calM'$ is identical to $\calM$ except for deleted outputs on muted transitions in $\calM'$,
and additional transitions coming from open ends in $\calM$.
Hence either $u\in\dom(\sem{\calM})$ but $u\in\dom(\sem{\calM'})$, in which case $u$ uses a transition from an open end
    in its run in $\calM$, or $\sem{\calM}(u)\neq\sem{\calM'}(u)$, in which case $u$ uses a muted transition
    in its run in $\calM'$.

$\calM''$ is identical to $\calM$ except for the output of the muted and open transitions.
Hence, when we run the counterexample $u$ in $\calM''$, we know it will use a muted or open transition,
and thus that the second statement is true.\qed
\end{proof}

\section{Finiteness proof for Theorem~\ref{thm:guar-tr}}

Theorem~\ref{thm:guar-tr} necessitates two thing to prove that Algorithm~\ref{proc:algoUpTransducer} terminates.
First, the number of equivalence queries must be bounded.
This is the core of the use of CompetingMinGen($T$) in Algorithm~\ref{proc:algoUpTransducer}
and what provides the improved guarantee of graybox learning.

Second, we need to ensure that the number of calls to CompetingMinGen($T$)
between two equivalence queries is bounded. To that effect, we show that for each number of states $n$,
the number of transducers that Algorithm~\ref{proc:algoUpTransducer} considers and eliminates is bounded.

\begin{proposition}\label{prop:AlgoSizenFinite}
Algorithm~\ref{proc:algoUpTransducer} considers finitely many transducers of size $n$.
\end{proposition}

\begin{proof}
Each membership test asked by CompetingMinGen($T$) eliminates a possible transducer compatible with $T$ in its current state.
What we will show is that for each number of states $n$, the number of transducers to eliminate is finite.
This works by considering a transducer as a series of finite choices.

\emph{States:} Since the equivalence relations of  MMs are right congruences,
there exists a prefix-closed set of $n$ elements, each belonging in a different state.
Since the number of finite prefix-closed sets of $\Sigma^*$ is finite, this choice is bounded.

\emph{Transitions without output:} The transitions (output excluded) are partial functions from $Q\times \Sigma$ to $Q$,
where $Q$ is the set of $n$ states.
Since the number of such functions is finite, this choice is bounded.

\emph{Open and muted transitions:} Each transition can be open, muted, or neither, depending on the state of $T$.
That leaves less than $3^D$ possibilities, where $D$ is the number of transitions. This choice is bounded.
This is not a tight analysis: adding observations can make a transition no longer open, or no longer muted,
but never the other way around, so we will never explore all $3^D$ possibilities.

\emph{Transition outputs and final outputs:} Open and muted transitions have a fixed output of $\varepsilon$.
Other transitions have an output bounded by rule $3$ of Definition~\ref{def:MergingMap}.
As we add observations, this bound gets more and more restrictive, never less.
The final outputs' existence and value are dictated by rule $5$ of Definition~\ref{def:MergingMap}.
Hence the choice of transition outputs and final outputs is bounded.

In conclusion, for every number of states $n$, we have a final number of transducers to eliminate,
and each membership query by CompetingMinGen($T$) eliminates one.
Since $n$ goes from $1$ to the size of the smallest transducer compatible with the target of our learning algorithm, Algorithm~\ref{proc:algoUpTransducer} terminates.
\end{proof}

\section{Small model theorems to bound witness sizes in CompetingMinGen($T$)}

The finiteness of counterexamples $u\in\Sigma^*$ such that $u\in\Up$ and
two competing candidates have a different image for $u$, or $u$ uses a muted or open transition,
is guaranteed by small model theorems on transducers.

\begin{lemma}\label{lem:DifferenceWitness}
Let $\calM_1, \calM_2$ be two transducers. 
If $\sem{\calM_1}_{|\Up}\neq\sem{\calM_2}_{|\Up}$,
then there is $u\in\Up$ s.t. $|u|\leq 2\cdot(|\calM_1|\cdot|\calM_2|\cdot|\calA_{\Up}|)^2$
and $\sem{\calM_1}$, $\sem{\calM_2}$ differ on $u$.
\end{lemma}

\begin{proof}
\abchanged{%
The case where $u\in\dom(\sem{\calM_1})\backslash\dom(\sem{\calM_2}$ (or vice-versa) is a simple case of difference between
the languages $\dom(\sem{\calM_1})\cap\Up$ described by $\calA_{\Up} \times \calM_1$
and $\dom(\sem{\calM_2})\cap\Up$ described by $\calA_{\Up} \times \calM_2$.
Such a difference is of size less than the size of the product of those two automata,
which is $2\cdot(|\calM_1|\cdot|\calM_2|\cdot|\calA_{\Up}|^2)$
(although the square on $|\calA_{\Up}|$ can easily be removed).

The case where $\sem{\calM_1}(u)\neq\sem{\calM_2}(u)$ is a classical result, provided in~\cite{schutzenberger1975relations}. }
Lemma 1 ensures that if the transducer $\calA_{\Up} \times (\calM_1\cup \calM_2)$ has a word with more than one output, there exists a $u\in\Sigma^*$ of size less than $1+2\cdot(|\calM_1|\cdot|\calM_2|\cdot|\calA_{\Up}|)\cdot(|\calM_1|\cdot|\calM_2|\cdot|\calA_{\Up}|-1)$ such that $u$ has more than one output.
\end{proof}

\begin{lemma}\label{lem:ArbitraryWitness}
Let $\calM$ be a transducer. 
If there exists $v\in\dom(\sem \calM_{|\Up})$ that uses a specific transition of $\calM$,
then there exists $u\in\dom(\sem \calM_{|\Up})$ of size less or equal than $2\cdot|\calA_{\Up}|\cdot|\calM|$
that uses this same transition.
\end{lemma}

\begin{proof}
This is a small model theorem on the transducer $\calA_{|\Up}\times \calM$.
We consider the path that leads $v$ to and from the targeted transition $q\xrightarrow{a|w}q'$ in $\calM$,
and cut $v$ into $v_0\cdot a\cdot v_1$,
where $v_0$ is the part of $v$ from the initial state to $q$,
and $v_1$ is the part of $v$ from $q'$ to a final state.
We can remove from $v_0$ and $v_1$ any loop of $\calM\cdot\calA_{|\Up}$,
to obtain $u_0$ and $u_1$ such that $|u_0|$ and $|u_1|$ are both strictly smaller than $|\calA_{|\Up}|\times|\calM|$, $u_0$ still goes from the initial state of $\calM$ to $q$, $u_1$ still goes from $q'$ to a final state of $\calM$, and $u=u_0\cdot a\cdot u_1$ is still accepted by $\calA_{|\Up}$.
Hence $u\in\Up$ and $u$ uses the targeted transition $q\xrightarrow{a|w}q'$ in $\calM$.

\end{proof}

\section{Effective computation of CompetingMinGen($T$)}\label{sec:TransducersUpEncoding}

We present a way to encode the problem of finding merging maps.
More specifically, to build the oracle CompetingMinGen($T$),
we need to see if, given a state number $n$, there exists
one merging map with two minimal resulting transducers (and its witness),
or two non-equivalent merging maps (and their witness).

We encode this problem as word equations, that is conjunctions and disjunctions of both equality and inequalities between words in a given domain~\cite{lin2018quadratic},\rbchanged{ with additional operators: concatenation and
membership in a finite set}.
SMT solvers have specialized tools~\cite{zheng2013z3,liang2016efficient} to solve such equations. 
As we consider a bounded domain of definition\rbchanged{ once we have fixed the size of transducers we are considering}, we could alternatively write these equations as a SAT query,\rbchanged{ as all the operators except the universal quantification can be polynomially encoded into SAT. Since we consider universal formulae over either $P_T$ or $\Sigma$, such formulae could also be encoded in SAT by replacing universal quantifiers by a conjunction on polynomially many variables in $P_T$ or $\Sigma$.}
\abchanged{This allows us to keep the theoretical NP bound of a SAT encoding,
with a more readable string equations encoding, more adapted to the problems we consider.}

For convenience, we say in what follows that a word variable $s$ can be undefined (denoted $s=\bot$).
In practice, we add a boolean variable $s^{def}$ associated to every word variable, that is true if and only if the associated variable is defined.

We now present equations for the various objects introduced in Section~\ref{sec:TransducersUp}:

\paragraph{Table:} A table $T: (P\cup P\Sigma)\cdot S\to \Gamma^*\cup\{\bot,\#\}$ is represented as the set that contains for all $u\in P_T$ 
    a boolean $T^{\#}_{u}$ that is true if and only if $T(u)= \#$, 
    and a word $T_{u}$ on $\Gamma^*$ that we assign to $T(u)$ if $T(u)$ is defined and $\bot$ otherwise.
\abchanged{We call $\phi_{T}$ the conjunction of assignations that describe the table $T$.}

\paragraph{Merging map:} A merging map $(\equiv,f)$ is represented as
 a boolean $E_{u,v}$ for all $u,v\in P_T$ representing if $u\equiv v$,
 and a word $f_u$ on $\Gamma^*$ for all $u\in P_T$ that represents the output of $f(u)$.
Every $f_u$ is of size smaller than $|u|\cdot\text{max}\lbrace|T(v)|\rbrace$.

The following equations ensure that $\equiv$ is an equivalence relation: 
For all $u,v,w\in P_T$, $E_{u,u}=\mathtt{true}$, $E_{u,v} = E_{v,u}$, and $E_{u,v} \wedge E_{v,w} \Rightarrow E_{u,w}$. 
\abchanged{We call $\phi_{eq}$ the conjunction of those constraints.}

We encode each point of Definition~\ref{def:MergingMap} in word equations. For all $u,v\in P_T$:
\begin{enumerate}
    \item $f_u=\bot \Rightarrow (f_v= \bot \Longleftrightarrow E_{u,v})$
    \item $(\exists u'\in P_T.\ T_{u\cdot u'}\neq \bot \wedge \neg T^{\#}_{u\cdot u'}) \Rightarrow (f_u\neq \bot \wedge \exists x\in \Gamma^*.\  f_u\cdot x=f_{u\cdot u'})$
    \item For all $a\in \Sigma$: $f_{ua}\neq \bot \Rightarrow (f_u\neq \bot \wedge \exists x\in \Gamma^*.\  f_u\cdot x=f_{ua})$
    \item For all $a\in \Sigma$ such that $ua,va\in P_T$: $(f_{u}\neq\bot\wedge E_{u,v}) \Rightarrow ( E_{ua,va}\wedge f_{ua}\neq \bot \Rightarrow (\exists x\in \Gamma^*.\ f_u\cdot x = T_u \wedge f_v\cdot x = T_v))$
    \item $(T_{u}\neq\bot\wedge \neg T^{\#}_{u} \wedge E_{u,v}) \Rightarrow (\neg T^{\#}_{u} \Rightarrow (\exists x\in \Gamma^*.\ f_u\cdot x = T_u \wedge f_v\cdot x = T_v))$
    \item $(\forall v\in P_T,\  va \notin P_\Gamma \vee \neg E_{u,v})\Rightarrow f_{ua}=\bot\lor f_{ua} = f_u$
\end{enumerate}
\abchanged{We call $\phi_{\text{mm}}$ the conjunction of those constraints and $\phi_{eq}$.}

\begin{lemma}\label{lem:EncodeMergingMap}
Let $T$ be an observation table. 
    If the variables $T(u)$, $E_{u,v}$, and $f_u$ respect \abchanged{$\phi_{T}\wedge\phi_{mm}$},
    then $(E_{u,v},f_u)$ represents a merging map of $T$.
\end{lemma}

\paragraph{Transducer run:} To encode the case of CompetingMinGen($T$)
where there exists two minimal merging map with non-equivalent corresponding transducers,
we first show how to encode that
a word $u\in \Sigma^\ell$ is accepted by the resulting transducer of a merging map with output $w$.
We define
 letters $a_1,\ldots, a_\ell\in\Sigma$ such that $u=a_1\cdot\ldots\cdot a_\ell$,
 input words $v_0,v_1,\ldots, v_\ell\in P_T$ that are representative of the equivalence classes of the states visited,
 output words $w_1,\ldots, w_\ell,w_{out}\in\Gamma^*$ such that $w=w_1\cdot\ldots\cdot w_n\cdot w_{out}$,
 and a word $\delta_{u,a}$ on $\Gamma^*$ that represents the output of transition $(u,a,ua)$.

    
The following equations ensure that $u$ is accepted by the transducer corresponding to the merging map, with output $w$: 
 $\forall u\in P_T,a\in \Sigma: f_u\cdot\delta_{u,a}=f_{ua}$
 $v_0=\varepsilon$, 
 $\forall i\in[1,n], E_{v_i,(v_{i-1}\cdot a_i)}$, 
 $\forall i\in[1,n], w_i=\delta_{v_{i-1},a_i}$, 
 and $f_{v_n}\cdot w_{out}= T_{v_n}$.
\abchanged{We call $\phi_{run}(u,w)$ the conjunction of those constraints.

Conversely, we can create a formula that ensures that a word $u$ fails in the transducer resulting from the MM
described by $E_{u,v}$ and $f_u$.
We still use $a_1,\ldots, a_\ell\in\Sigma$ such that $u=a_1\cdot\ldots\cdot a_\ell$,
and input words $v_0,v_1,\ldots, v_\ell\in P_T$ that are representative of the equivalence classes of the states visited,
except we now allow $v_i$ to be undefined ($v_i=\bot$).
We do not care about the output.
    
The following equations ensure that $u$ fails in the transducer corresponding to the merging map,
by failing to be able to read a letter or ending in a non-final state: 
$v_0=\varepsilon$, 
$\forall i\in[1,n], E_{v_i,(v_{i-1}\cdot a_i)}\lor v_i=\bot$, 
and $v_n=\bot\lor \forall v\in P_T, E_{v,v_n}\implies T(v)=\bot$.
We call $\phi_{\lnot{run}}(u)$ the conjunction of those constraints.
}

\begin{lemma}\label{lem:EncodeResultingTransducerRun}
    Let $T$ be an observation table. 
    If $(E_{u,v},f_u)$ represents a merging map of $T$,
    and $u\in\Sigma^*$, $w\in\Gamma^*$ \rbchanged{respect $\phi_{run}(u,w)$},
    then $u$ is accepted by the resulting transducer of this merging map with output $w$.
\end{lemma}

We can, using classical methods, encode that a word accepted by a regular language $\Up$
is accepted by two different merging maps with different output. \abchanged{We call this formula $\phi_{\Up}(u)$}.
Thus:

\begin{lemma}\label{lem:EncodeDifferentMergingMap}
    Let $T$ be an observation table. 
    It can be encoded with word equations there exists two merging maps of $T$ \abchanged{of size $n$}
    and corresponding transducers \abchanged{non-equivalent on $\Up$}.
    If this is the case, a witness $u$ can be returned.
\end{lemma}

\begin{proof}
Using several sets of word equations, what we want is:
\begin{enumerate}
    \item a merging map represented by the variables $E,f$ \abchanged{of size $n$}
    \item an input-output pair $u,w$ of its resulting transducer, where $|u|\leq\ell$
    \item a merging map represented by the variables $E',f'$ \abchanged{of size $n$}
    \item \begin{itemize}
        \item that $u$ fails its run in the resulting transducer of $E',f'$, or
        \item an input-output pair $u',w'$ of its resulting transducer, $|u'|\leq\ell$
        and that $u=u'\in\Up$ but $w\neq w'$,
    \end{itemize}
\end{enumerate}

The four first points have been proven in Lemma~\ref{lem:EncodeMergingMap} and~\ref{lem:EncodeResultingTransducerRun}.
\abchanged{To create two merging maps we use variables $E_{u,v}$ and $f_u$ for the first one, and $E'_{u,v}$ and $f'_u$ for the second.
We call $\phi'_{mm}$, $\phi'_{run}(u,w)$ and $\phi'_{\lnot run}(u,w)$ the versions of
$\phi_{mm}$, $\phi_{run}(u,w)$ and $\phi_{\lnot run}(u,w)$
that use the variables $E'_{u,v}$ and $f'_u$ instead of $E_{u,v}$ and $f_u$.}
If we want to find two different outputs for $u$, the resulting formula is:
$$\Phi_{wit1}=
\phi_{T}\wedge\wedge\phi_{mm}\wedge\phi_{run}(u,v)
\wedge\phi'_{mm}\wedge\phi'_{run}(u',v')
\wedge u=u'\wedge \phi_{\Up}(u)\wedge w\neq w$$
If we want to find $u$ that has a run in the first transducer but not the second, the resulting formula is:
$$\Phi_{wit1}=
\phi_{T}\wedge\wedge\phi_{mm}\wedge\phi_{run}(u,v)
\wedge\phi'_{mm}\wedge \phi'_{\lnot run}(u)
\wedge *\phi_{\Up}(u)$$
The length $\ell$ to limit the size of $u$ comes from Lemma~\ref{lem:DifferenceWitness}.
\end{proof}

\paragraph{Merged and open transitions:}
It now remains to encode the case where a single merging map uses muted or open transitions.
As shown in Proposition~\ref{prop:CompetingNSC}, the output of the witness does not matter
as much as finding a word $u\in\Up$ that uses one of those transitions.
We define 
    letters $a_1,\ldots, a_n\in\Sigma^n$ such that $u=a_1\cdot\ldots\cdot a_n$,
    words $v_0,\ldots, v_n\in P_T$ that are representatives of the equivalence classes of the states visited,
    booleans $m_{u,a}$ with $u\in P_T$ and $a\in \Sigma$ representing whether $u$ reading $a$ is 
    a muted transition, as specified in point 6 of Definition~\ref{def:MergingMap}, 
    and booleans $o_{u,a}$ with $u\in P_T$ and $a\in \Sigma$ representing whether $u$ reading $a$ is
    an open transition, as specified by their definition.


$m_{u,a}$ and $o_{u,a}$ are encoded as follows:
 $m_{u,a}\iff(\forall v\in P_T,\ va\notin P_\Gamma\lor u\not\equiv v)
    \land(\exists v\in P_T,\ u\equiv v\land f(va)\neq\bot)$, 
 and $o_{u,a}\iff(\forall v\in P_T,\ va\notin P_T\lor u\not\equiv v)$


As before, we know how to encode in equations that a word $u$ is accepted by $\Up$.
Adding the following equations ensures that
this $u$ goes through an open or muted transition
in an open completion of the transducer resulting from the merging map.
As seen in the proof of Proposition~\ref{prop:CompetingNSC},
this is sufficient to prove that there exists several (and in fact infinitely many) transducers compatible
with this merging map that would write different outputs when reading $u$: 
    $v_0=\varepsilon$,
    $\forall i\in [1,\ell], f_{v_i}\neq\bot$, 
    $\forall i\in [1,\ell],\ E_{v_i,(v_{i-1}\cdot a_i)}\lor o_{v_{i-1},a_i}$,
    $\forall i,j\in[1,\ell],\ E_{v_i,v_j}\Rightarrow E_{v_{i+1},v_{j+1}}$, 
    $T(v_\ell)\neq\bot$, 
    and $\exists i\in [1,\ell],\ o_{v_{i-1},a_i}\lor m_{v_{i-1},a_i}$.
\abchanged{We call $\phi_{mo}(u)$ the conjunction of those constraints and the definitions of $m_{u,a}$ and $o_{u,a}$.}

\begin{lemma}\label{lem:EncodeArbitraryTransition}
    Let $T$ be an observation table, if $(E_{u,v},f_u)$ represents a merging map of $T$, \abchanged{and $u\in\Sigma^*$ respects $\phi_{mo}(u)$}, 
    then $u$ uses a muted or open transition in an open completion of the transducer resulting from the merging map.
\end{lemma}

\begin{proof}
The conditions above describe a path in the transducer resulting from the merging map,
with three additional conditions:
\begin{itemize}
    \item $\forall i\in [1,\ell],\ E_{v_i,(v_{i-1}\cdot a_i)}\lor o_{v_{i-1},a_i}$
    allows for a path to continue from an open end ($o_{v_{i-1},a_i}$) when no transition is available.
    \item $\forall i,j\in[1,\ell],\ E_{v_i,v_j}\Rightarrow E_{v_{i+1},v_{j+1}}$
    ensures that reading a certain $a$ in a certain state will always lead to the same successor. This is necessary as transitions coming from an open completion are not constrained by the equations defining the merging map.
    \item $\exists i\in [1,\ell],\ o_{v_{i-1},a_i}\abchanged{\lor m_{v_{i-1},a_i}}$
    ensures that $u$ uses at least one muted or open transition
\end{itemize}
\end{proof}

\begin{lemma}\label{lem:EncodeArbitraryWitness}
    Let $T$ be an observation table. 
    It can be encoded with word equations there exists a merging map of $T$ \abchanged{of size $n$}
    and a word $u\in\Up$ such that $u$ uses a muted or open transition in an open completion of the transducer resulting from this merging map.
\end{lemma}

\begin{proof}
Using several sets of word equations, what we want is a merging map represented by the variables $E,f$ \abchanged{of size $n$} and an input word $u\in\Up$ such that $|u|\leq\ell$,
    and $u$ uses a muted or open transition in an open completion of the transducer resulting from the merging map.

Encoding $u\in\Up$ is possible through classical results.
Lemma~\ref{lem:EncodeArbitraryTransition} ensures us that it is possible to encode those points, \rbchanged{with the resulting formula being $\Phi_{wit2}=\phi_{T}\wedge\phi_{mm}\wedge\phi_{mo}(u)\wedge\phi_{\Up}(u)$.}
The length $\ell$ to limit the size of $u$ comes from Lemma~\ref{lem:DifferenceWitness}.
\end{proof}

We can thus conclude with the following proposition:

\begin{proposition}\label{prop:EncodeCompetingMinGen}
    Let $T$ be an observation table.
    The subroutines MinGen($T$) and CompetingMinGen($T$) can be effectively implemented.
\end{proposition}

\begin{proof}
Lemma~\ref{lem:EncodeMergingMap},~\ref{lem:EncodeDifferentMergingMap} and~\ref{lem:EncodeArbitraryWitness}
ensure that we can decide the existence of a merging map, two competing merging maps,
or a word using a muted or open transition, if given a number $n$ of states.
Notably, those last two points combine thanks to ooposition~\ref{prop:CompetingNSC} to guarantee that if
two competing transducers of size $n$ exist, then we can find them that way.

To find the proper number of state $n$, we try to guess a transducer with $i$ states,
$i$ starting at $1$ and incremented each time we cannot find a transducer.
\end{proof}



\end{document}